\documentclass[nohyperref]{article}

\usepackage{microtype}
\usepackage{graphicx}
\usepackage{subfigure}
\usepackage{booktabs} % for professional tables

\usepackage{hyperref}

\usepackage[accepted]{icml2022}

\usepackage{amsmath}
\usepackage{amssymb}
\usepackage{mathtools}
\usepackage{amsthm}

\usepackage[capitalize,noabbrev]{cleveref}

\theoremstyle{plain}
\newtheorem{theorem}{Theorem}[section]
\newtheorem{proposition}[theorem]{Proposition}
\newtheorem{lemma}[theorem]{Lemma}
\newtheorem{corollary}[theorem]{Corollary}
\theoremstyle{definition}
\newtheorem{definition}[theorem]{Definition}

\theoremstyle{remark}

\usepackage[textsize=tiny]{todonotes}

\icmltitlerunning{Impossibility of Learning to Cooperate with Adaptive Partners}

\usepackage{algorithm}
\usepackage{algorithmic}
\usepackage{verbatim}

\DeclareMathOperator*{\argmin}{arg\,min}
\DeclareMathOperator*{\argmax}{arg\,max}

\theoremstyle{definition}
\newtheorem{example}{Example}

\begin{document}

\twocolumn[
\icmltitle{On the Impossibility of Learning to Cooperate with Adaptive Partner Strategies in Repeated Games}

% It is OKAY to include author information, even for blind
% submissions: the style file will automatically remove it for you
% unless you've provided the [accepted] option to the icml2022
% package.

% List of affiliations: The first argument should be a (short)
% identifier you will use later to specify author affiliations
% Academic affiliations should list Department, University, City, Region, Country
% Industry affiliations should list Company, City, Region, Country

% You can specify symbols, otherwise they are numbered in order.
% Ideally, you should not use this facility. Affiliations will be numbered
% in order of appearance and this is the preferred way.
\icmlsetsymbol{equal}{*}

\begin{icmlauthorlist}
\icmlauthor{Robert Loftin}{delft}
\icmlauthor{Frans A. Oliehoek}{delft}
\end{icmlauthorlist}

\icmlaffiliation{delft}{Department of Intelligent Systems, Delft University of Technology, Delft, South-Holland, The Netherlands}

\icmlcorrespondingauthor{Robert Loftin}{R.T.Loftin@tudelft.nl}

% You may provide any keywords that you
% find helpful for describing your paper; these are used to populate
% the "keywords" metadata in the PDF but will not be shown in the document
\icmlkeywords{Interactive Learning, Online Learning, Human-AI Interaction, Coordination, Cooperation, Reinforcement Learning}

\vskip 0.3in
]

% this must go after the closing bracket ] following \twocolumn[ ...

% This command actually creates the footnote in the first column
% listing the affiliations and the copyright notice.
% The command takes one argument, which is text to display at the start of the footnote.
% The \icmlEqualContribution command is standard text for equal contribution.
% Remove it (just {}) if you do not need this facility.

\printAffiliationsAndNotice{}  % leave blank if no need to mention equal contribution
%\printAffiliationsAndNotice{\icmlEqualContribution} % otherwise use the standard text.

\begin{abstract}
Learning to cooperate with other agents is challenging when those agents also possess the ability to adapt to our own behavior.  Practical and theoretical approaches to learning in cooperative settings typically assume that other agents' behaviors are stationary, or else make very specific assumptions about other agents' learning processes.  The goal of this work is to understand whether we can reliably learn to cooperate with other agents without such restrictive assumptions, which are unlikely to hold in real-world applications.  Our main contribution is a set of impossibility results, which show that no learning algorithm can reliably learn to cooperate with all possible adaptive partners in a repeated matrix game, even if that partner is guaranteed to cooperate with \emph{some} stationary strategy.  Motivated by these results, we then discuss potential alternative assumptions which capture the idea that an adaptive partner will only adapt \emph{rationally} to our behavior.
\end{abstract}

\section{Introduction}
\label{introduction}

Artificial agents deployed in real-world settings must be able to reliably interact and cooperate with humans (as well as other artificial agents).  While machine learning can allow such agents to adapt to others' behavior, learning may fail when other agents are equally capable of adapting to our agent's own behavior. In this work, we consider whether an agent can reliably learn to cooperate with other adaptive agents when the strategies of those agents are unknown, and the minimal assumptions about others' strategies needed to ensure such cooperation.  We show that, even when another agent's strategy is known be ``cooperative'' in a reasonable sense, there is no algorithm which is guaranteed to learn to cooperate without additional assumptions.  At the same time, we argue that the assumptions under which existing approaches can reliably learn to cooperate~\cite{farias2003experts,even2005pomdp} are too restrictive for many realistic settings, and suggest alternative constraints that may be more easily satisfied in practice.

We focus on the case where two agents, who we will refer to as \emph{Alice} and \emph{Bob}, must work together to achieve some common goal.  Our challenge is to choose a ``good'' strategy for Alice, without knowledge of Bob's strategy.  A common way to approach such problems is to assume that Bob will play his half of some optimal joint strategy, and so have Alice play her half of this strategy as well.  This \emph{objective} approach may fail, however, when the task admits multiple, incompatible solutions, and Alice has no way to know in advance which solution Bob has chosen.  For this reason, some works~\cite{mescheder2011pomdp,poland2006aixi} have advocated framing the problem as one of single-agent learning, where Alice's strategy adapts to Bob's observed behavior.  This \emph{subjective} approach becomes problematic, however, when Bob's strategy can also adapt to Alice's behavior.  The key theoretical result of this work in Theorem~\ref{thm:impossibility}, which shows that if Bob's strategy is allowed to be adaptive, then there is \emph{no strategy} Alice can choose that will reliably cooperate with every strategy Bob could follow.

More specifically, we show that for some infinitely repeated matrix games, there is no adaptive strategy Alice can follow that will always achieve the same average payoff as the best strategy she could have followed given knowledge of Bob's strategy.  The key idea is that as the game progresses, Bob may continuously increase the time over which Alice must commit to a single strategy before his behavior converges, such that Alice is never able to gain information about the long-term payoffs of potential strategies.  This stands in contrast to existing results, which guarantee near-optimal payoff under the assumption that there exists some bounded time horizon over which Alice can evaluate possible strategies against Bob's behavior.  We will see that this assumption is violated even by seemingly reasonable learning strategies, such as fictitious play~\cite{robinson1951fictitious} (see Section~\ref{open_ended_strategies}), that Bob could be expected to follow.

In Section~\ref{open_ended_strategies}, we introduce the concept of an \emph{open-ended} strategy, which weakens the assumption of a bounded evaluation horizon, and formalize the problem of no-regret learning against an adaptive, open-ended partner.  In Sections~\ref{passive_algorithms} and~\ref{active_algorithms}, we prove results for two restricted classes of adaptive strategies Alice could follow, while in Section~\ref{general_case} we prove our main impossibility result.  Finally, in Section~\ref{discussion} we discuss potential alternatives to bounded-horizon assumptions, based on the idea that adaptive elements of Bob's strategy must be \emph{rational} with respect to the shared goal.

\section{Preliminaries}
\label{preliminaries}

We restrict ourselves to the case where Alice must learn to cooperate with an unknown partner (Bob) in an infinitely repeated matrix game.  Let $A$ and $B$ be the \emph{finite} action sets for Alice and Bob respectively.  With a slight abuse of notation, we will denote by $G$ both the $\vert A \vert \times \vert B\vert$ payoff matrix of the stage game, and the infinitely repeated game defined by $G$.  Let $G(a,b) \in [0,1]$ be the payoff Alice receives for taking action $a$ when Bob takes action $b$.  Let $a_n$ and $b_n$ be the actions taken by Alice and Bob in stage $n$, such that $G(a_n, b_n)$ is the payoff received by Alice in stage $n$.  We assume that Alice knows the values of the payoff matrix $G$, and only needs to learn about Bob's strategy.  Our main results assume that Bob's strategy can be chosen arbitrarily, and so we will not explicitly consider Bob's payoffs.

Let $\mathcal{H}_n = (A \times B)^n$ be the set of histories of play of length $n$, and let $\mathcal{H} = \bigcup^{\infty}_{n=0}\mathcal{H}_n$ be the set of all \emph{finite} histories.  Let $\vert h\vert$ be the length of $h$, and let $a(h)$ be Alice's most recent action in $h$.  We overload $h$ to also denote the \emph{event} that history $h \in \mathcal{H}$ occurs during play.  We represent the strategies followed by Alice and Bob as \emph{behavioral strategies}~\cite{shoham2008multiagent}, that is, as functions $\pi: \mathcal{H} \mapsto \Delta(A)$ and $\phi : \mathcal{H} \mapsto \Delta(B)$, where $\Delta(X)$ is the set of probability distributions over $X$.  When Alice implements a learning algorithm $\mathcal{A}$, we let $\mathcal{A}_{G}$ denote the strategy she would follow when playing a specific game $G$.  We let $\mathcal{A}_{G,E}$ denote Alice's strategy when $\mathcal{A}$ additionally depends on a set $E$ of \emph{expert} strategies.  Finally, we assume Alice and Bob's respective strategies are defined for all histories $h \in \mathcal{H}$, including histories that occur with probability zero.

We evaluate strategies based on the infinite-time limit of Alice's expected per-stage payoff.  When Alice follows strategy $\pi$, while Bob follows strategy $\phi$, for history $h \in \mathcal{H}$ we denote Alice's conditional expected average payoff as
\begin{equation}
    \label{eqn:value}
    V_{G}(\pi, \phi \vert h) \! = \! \liminf_{s \rightarrow \infty} \text{E}_{\pi,\phi}\!\!\left[ \frac{1}{s}\!\!\!\ \sum^{\vert h \vert + s}_{n = \vert h\vert + 1}\!\!\! G(a_n, b_n) : h \right].
\end{equation}
For convenience, we let $V_{G}(\pi, \phi) = V_{G}(\pi, \phi \vert \emptyset)$, where $\emptyset$ is the initial (empty) history. The limit inferior is used because the limit of the average payoff may not always exist. 

We define all probabilities and expectations with respect to a common probability space $(\Omega, \mathcal{F})$, where $\Omega = (A \times B)^*$, the set of all \emph{infinite} paths of play.  The history $h_n$ up to stage $n$ is a random variable w.r.t. $(\Omega, \mathcal{F})$, that is, $h_n : \Omega \mapsto \mathcal{H}_n$.  For all $n \in \mathbb{N}$, we let $\mathcal{F}_n \subset \mathcal{F}$ be the smallest $\sigma$-algebra containing all sets $\{\omega \in \Omega : h_n(\omega) = h\}$ for all $h \in \mathcal{H}_n$.  $\mathcal{F}_0, \mathcal{F}_1,\ldots$ is then the natural filtration~\cite{kallenberg1997probability} of the stochastic process $h_0,h_1,\ldots$.  We let $\text{Pr}_{\pi, \phi}\{\xi \;\big\vert\; h\}$, for $\xi \in \mathcal{F}$, denote the probability measure induced by $\pi$ and $\phi$ on $(\Omega, \mathcal{F})$, conditioned on $h \in \mathcal{H}$.  We will omit one or both strategies when they are clear from the context.  % Do we ever actually use the definition of the filtration?

\section{Related Work}
\label{related_work}

Existing frameworks for learning to cooperate with other agents make a wide variety of prior assumptions about those agents' strategies.  Unsurprisingly, these frameworks are forced to make trade-offs between computational and sample efficiency, and robustness to different partners.  At one end of the spectrum lie recent, practical approaches based on deep reinforcement learning, which assume that other agents will act (near-)optimally for the current task.  Under this assumption, the most straightforward approach is to choose a strategy that approximates one component of some optimal joint solution~\cite{foerster2019bad,hu2020Ssad} to this task.  As this may fail when there are multiple, incompatible solutions, some approaches select strategies that perform well against synthetic distributions over near-optimal strategies~\cite{canaan2019diverse,strouse2021collaborating,cui2021klevel}, or against learned models of real agents~\cite{carroll2019overcooked}.  Other techniques attempt to address specific sources of miscoordination, such as task symmetries~\cite{hu2020otherplay}, asymmetric roles~\cite{genter2011role}, and arbitrary conventions~\cite{hu2021obl}.

More general (though less scalable) is the Bayesian approach, modelling the interaction from a single agent's perspective as a POMDP, where the strategies of other agents are components of the latent state~\cite{macindoe2012pomcop,panella2016ipomdp}.  In the rational learning framework~\cite{kalai1993rational}, each agent follows a Bayes-optimal strategy given some prior over the other agent's strategy.  The related \emph{targeted learning} criteria~\cite{powers2004targeted, powers2005finite} simply require that an agent eventually learn to play optimally against any member of some target class of strategies.  A major limitation of these approaches is the difficulty in specifying priors (or target classes) for all agents that admit other agents' own learning strategies.  For example, in rational learning, Alice's prior must assign positive probability to Bob's true belief over Alice's strategy, which in turn depends on Alice's prior itself, leading to a circular dependence that is difficult to satisfy.

At the other end of the spectrum are approaches which treat other agents as simply being part of the (non-stationary) environment itself, and apply single-agent reinforcement learning within the interaction.  These include methods that model the problem as one of learning in a non-episodic POMDP~\cite{mescheder2011pomdp}.  Learning in POMDPs without resets is possible by evaluating potential strategies over an ever-expanding time horizon~\cite{even2005pomdp}.  Similarly, the $\epsilon$-greedy \emph{strategic experts algorithm}~\cite{farias2003experts} and the exponentially weighted \emph{follow or explore} algorithm~\cite{poland2005universal} evaluate expert strategies over expanding horizons.  As we discuss in Section~\ref{flexible_strategies}, all these approaches assume that even when other agents are adaptive, there is still some fixed, or slowly growing time horizon over which strategies can be evaluated. As with targeted and rational learning, the challenge is that these assumptions are generally not satisfied when other agent's implement equally capable learning algorithms.

\section{Adaptive Regret}
\label{adaptive_regret}

A key conceptual challenge in multi-agent learning is determining what constitutes a ``good'' strategy for one agent to follow given our assumptions about the strategies of other agents.  In this work, we implicitly assume that Bob's strategy will be cooperative, and that there is some corresponding strategy Alice could follow which would yield satisfactory, if not optimal, average payoff.  Therefore, a natural choice for Alice is to follow a \emph{no-regret} adaptive strategy that is guaranteed to eventually do as well in expectation as the best fixed strategy Alice could have chosen.  As we are primarily concerned with the case where Bob's strategy is also adaptive, we need a notion of regret that captures Bob's dependence on Alice's previous actions. 
\begin{definition}
    \label{def:adaptive_regret}
    Let $\pi$ and $\phi$ be the strategies followed by Alice and Bob respectively in a repeated game $G$, and let $E$ be a set of alternative \emph{expert} strategies for Alice.  Alice's \emph{adaptive regret}\footnote{\cite{dekel2012policy} provide a related notion of \emph{policy regret} over finite time horizons for deterministic $\phi$.} with respect to the set $E$ is then:
    \begin{equation}
        \label{eqn:adaptive_regret}
        \mathcal{R}^{\text{adapt}}_{G,E}(\pi, \phi) = \max_{e \in E}V_{G}(e, \phi) - V_{G}(\pi, \phi),
    \end{equation}
\end{definition}
The adaptive regret compares the expected average payoff Alice receives when following $\pi$ against the expected average payoff she could have received had she followed the best expert strategy $e \in E$ from the beginning of the game repeated game.  This is in contrast to the more common \emph{external regret}~\cite{hannan1957consistency}, defined as
\begin{equation}
    \label{eqn:external_regret}
    \mathcal{R}^{\text{ext},s}_{G,E} \! = \! \max_{e \in E} \text{E}_{\pi, \phi}\!\!\left[ \frac{1}{s}\!\sum^{s}_{n=1}\! G(e, b_n)\! - G(a_n, b_n) \right],
\end{equation}
which assumes that Bob's actions would remain unchanged had Alice always followed strategy $e$ rather than $\pi$.  Note that, when each $e \in E$ simply takes some fixed action at every stage, we overload $e$ to denote this action as well.  

As Bob's strategy is initially unknown, our goal is to choose a strategy $\pi$ for Alice that minimizes her \emph{worst-case} adaptive regret over all strategies $\phi$ that Bob could follow.  For some sufficiently constrained (e.g., finite) set of experts $E$, we might hope to find an \emph{adaptive} strategy $\pi^*$ which guarantees that Alice suffers no adaptive regret, that is:
\begin{equation}
    \sup_{\phi}\mathcal{R}^{\text{adapt}}_{G,E}(\pi^*, \phi) \leq 0.
\end{equation}
The following example, however, shows that Bob's ability to simultaneously adapt to Alice's behavior as well means that such a strategy $\pi^*$ will not always exist.

\begin{example}
    \label{exp:grim_trigger}
    Consider the game $G$ defined by the matrix
    \renewcommand{\arraystretch}{1.2}
    \begin{center}
        \begin{tabular}{|c|c|c|c|} \hline
                  & $b^1$ & $b^2$ & $b^3$ \\ \hline
            $a^1$ & $2$ & $0$ & $1$ \\ \hline
            $a^2$ & $0$ & $2$ & $1$ \\ \hline
        \end{tabular}
    \end{center}
    and Bob's \emph{grim trigger} strategies $\phi_1$ and $\phi_2$.  $\phi_1$ plays $b^1$ so long as the Alice plays $a^1$, and switches to $b^3$ for all future stages if the she ever plays $a^2$.  $\phi_2$ does the opposite, playing $b^2$ so long as Alice plays $a^2$, and otherwise switching to $b^3$.  Let $E = \{ a^1, a^2 \}$, where $a^1$ and $a^2$ are overloaded to denote the fixed-action strategies that play actions $a^1$ and $a^2$.  Observe that $V_{G}(a^1, \phi_1) = 2$ and $V_{G}(a^2, \phi_2) = 2$, because in both cases Bob never switches to action $b^3$.  There is, however, no strategy $\pi$ which simultaneously achieves $V_{G}(\pi, \phi_1) = 2$ and $V_{G}(\pi, \phi_2) = 2$.  If Alice plays $a^1$ at any point, then $V_{G}(\pi,\phi_2) = 1$, but if she never plays $a^1$, then $V_{G}(\pi, \phi_1) = 1$.  The best Alice can do is to select between strategies $a^1$ and $a^2$ with equal probability, such that $V_{G}(\pi, \phi_1) = V_{G}(\pi, \phi_2) = \frac{3}{2}$, and Alice always suffers an adaptive regret of $2 - \frac{3}{2} \geq 0$. \footnote{Similar grim-trigger examples were provided in~\citep{poland2005universal} and~\cite{dekel2012policy}.}
\end{example}

While such grim-trigger strategies may seem unrealistic in a cooperative setting, it \emph{is} reasonable to imagine that an agent might switch to a sub-optimal ``safety'' strategy if their partner initially fails to behave as they had expected.  To address this issue from a theoretical perspective, we can either: 1) explicitly rule out such strategies, or 2) use a performance metric that does not penalize Alice for early mistakes, even when these mistakes have permanent consequences.  In this section we focus on the former approach, leaving discussion of the later formulation to Section~\ref{open_ended_regret}.

\subsection{Flexible Strategies}
\label{flexible_strategies}

As discussed in Section~\ref{related_work}, there are several learning algorithms which Alice could implement to try to minimize her adaptive regret.  The regret guarantees provided for these algorithms share a common assumption that there is some bounded time horizon over which each expert can be reliably evaluated.  Concretely, the \emph{strategic experts algorithm}~\cite{farias2003experts} can guarantee no adaptive regret under the assumption that Bob's strategy is \emph{flexible} with respect to $E$ in the game $G$.

\begin{definition}
\label{def:flexibility}
Bob's strategy $\phi$ is \emph{flexible} w.r.t. $E$ in game $G$ if there exist constants $\mu_e, \forall e \in E$, $r > \frac{1}{4}$ and $c$ such that, for all histories $h \in \mathcal{H}$
\begin{equation}
    \label{eqn:flexibility}
    \text{E}_{a,b \sim e,\phi}\left[\;\; \left\vert \frac{1}{s}\!\!\!\!\sum^{\vert h\vert + s}_{n = \vert h \vert + 1}\!\!\!\! G(a_n, b_n) - \mu_e \right\vert \; : \; h \right] \leq cs^{-r}
\end{equation}
for each expert $e \in E$.
\end{definition}

Flexibility implies that no matter what actions Alice has taken in the past, there is some $\tau$ such that if Alice follows $e \in E$ for $\tau$ stages, $e$'s average payoff over this period will be a good estimate of the limit of $e$'s average payoff.  Furthermore, flexibility means that this $\tau$ is fixed, and independent of the history $h$.  We can see that the grim-trigger strategies from Example~\ref{exp:grim_trigger} are not flexible, as the value of $\mu_e$ would depend on whether the history $h$ included the triggering condition.  Similar assumptions of a bounded \emph{evaluation horizon} have been used to provide no-regret guarantees for non-episodic POMDPs~\cite{even2005pomdp} and against strategies that depend on only Alice's last $k$ actions~\cite{powers2005finite, dekel2012policy}.  The \emph{follow or explore} algorithm~\cite{poland2005universal} does allow us to weaken this assumption slightly, guaranteeing no-regret when the evaluation horizon grows slower than $O(\vert h\vert^{\frac{1}{4}})$.  % Is this bound provided in the 2005 or the 2006 paper?

\subsection{Open-Ended Strategies}
\label{open_ended_strategies}

The problem with flexibility (Definition~\ref{def:flexibility}) and related assumptions is that many learning algorithms fail to satisfy them, because their behavior depends on the full history of play.  As an example, consider fictitious play~\cite{robinson1951fictitious}, which plays a best response to the empirical distribution over its partner's previous actions. 
% Fictitious-Play Example
\begin{proposition}
    \label{prop:fictitious_play_flexible}
    Fictitious play is not flexible in all fully cooperative repeated games.
\end{proposition}

\begin{proof}
    Let $G$ be the repeated game from Example~\ref{exp:grim_trigger}.  If Bob's strategy $\phi^{\text{fp}}$ implements fictitious play, then from any history $h$, $V_{G}(a^2, \phi^{\text{fp}} \vert h) = 2$.  For any positive $c$, $r$, let Alice first play action $a^1$ for $n = 2(2c)^{\frac{1}{r}}$ stages.  If Alice plays $a^2$ for the next $\frac{n}{2}$ stages, she will receive an average payoff of $0$ over this period, since Bob will not switch to $b^2$ until Alice plays $a^2$ for $n+1$ steps.  The error in Alice's estimate of $a^2$'s average long-term payoff will be $2$.  On the other hand, the error bound $c(\frac{n}{2})^{-r} = \frac{1}{2}$.  Therefore, no constants $c$, $r$ can satisfy Equation~\ref{eqn:flexibility} for all $h$, and so $\phi^{\text{fp}}$ is not flexible.
\end{proof}

Assuming that Bob's strategy is flexible therefore rules out the possibility that Bob is using fictitious play (at least with unbounded memory), and indeed rules out many natural learning algorithms that Bob might implement.  We therefore consider whether it is possible to relax the assumption of a fixed evaluation horizon, and still guarantee that Alice achieves optimal average payoff in the limit.  Rather than assuming that Bob's strategy is flexible with respect a set of experts $E$, we make the weaker assumption that Bob's strategy is \emph{open-ended} with respect to $E$.

\begin{definition}
    \label{def:open_ended}
    Bob's strategy $\phi$ is \emph{open-ended} w.r.t. experts $E$ in game $G$ if, $\forall e \in E$, there exists $\mu_e$ such that
    \begin{equation}
        \label{eqn:open_ended}
        V_{G}(e,\phi \vert h) = \mu_e
    \end{equation}
    for all histories $h \in \mathcal{H}$.
\end{definition}

The following propositions (proven in Appendices~\ref{apx:open_ended_flexible} and~\ref{apx:fictitious_play_open_ended}) show that being open-ended is a strictly weaker condition than flexibility, and that it allows for learning strategies which are not always flexible.

\begin{proposition}
    \label{prop:open_ended_flexible}
    In any repeated game $G$, if Bob's strategy is flexible with respect to a set of experts $E$, then it is also open-ended w.r.t. $E$.
\end{proposition}

\begin{proposition}
    \label{prop:fictitious_play_open_ended}
    For all fully cooperative, finite-action repeated games $G$, fictitious play is open-ended w.r.t. the set of all fixed-action strategies.
\end{proposition}

When Bob's strategy is open-ended, if at any point Alice commits to following a single expert $e \in E$ for all future stages, then she is guaranteed to receive an average payoff of $\mu_e$ in expectation.  Our main result, however, shows that this fact alone does not allow for an algorithm which is guaranteed to do as well in expectation as the best expert in $E$ against all possible open-ended strategies.

\begin{theorem}
    \label{thm:impossibility}
    For any $N \geq 3$, there exists an $N \times N$ repeated game $G$, with a finite set of experts $E$, such that, for any learning algorithm $\mathcal{A}$, and any $0 < \delta < \frac{N-2}{N}$, we can construct a strategy $\phi$ such that
    \begin{equation}
        \label{eqn:impossibility}
        \mathcal{R}^{\text{\emph{adapt}}}_{G,E}(\mathcal{A}_{G,E}, \phi) \geq \frac{1}{2}\left[\frac{N-2}{N} - \delta \right]^2
    \end{equation}
    and such that $\phi$ is open-ended w.r.t. $E$ in $G$.
\end{theorem}

Theorem~\ref{thm:impossibility}, which we will prove in Section~\ref{general_case}, means that in some games, no matter what strategy Alice follows, there will always be an open-ended strategy Bob can choose such that Alice's expected average payoff is less than that of the best expert.  Furthermore, for a large enough game and small enough value of $\delta$, Alice's regret can be made arbitrarily close to $\frac{1}{2}$ (recall that payoffs are bounded in $[0,1]$).  Here $\delta$ is a parameter of the strategy $\phi$ that can be chosen arbitrarily, and is an artifact of our proof technique.

\section{Impossibility Results}
\label{impossibility}

To prove Theorem~\ref{thm:impossibility}, it will be sufficient to show that no-regret learning against an open-ended strategy is impossible in an $N \times N$ coordination game $G$, where $G(a, b) = 1$ if $a = b$ and $G(a, b) = 0$ otherwise.  The set $E$ will consist of $N$ fixed-action strategies, where again with a slight abuse of notation we denote by $e \in E$ the expert that in each stage takes action $e \in \{1,\ldots,N\}$.  We will demonstrate that for any algorithm $\mathcal{A}$ that Alice could choose, it is always possible for Bob to choose a strategy that is open-ended, but nonetheless prevents Alice's strategy $\mathcal{A}_{G,E}$ from doing as well as the best expert $\bar{e}$ in $E$.  Before proving Theorem~\ref{thm:impossibility}, we will first prove that that no-regret learning is impossible for two special classes of algorithms.  These are algorithms that are either \emph{passive}, and eventually commit to following a fixed expert, or \emph{active}, and never stop exploring.

\subsection{Passive Algorithms}
\label{passive_algorithms}

By definition, when Bob's strategy is open-ended, if Alice commits to a fixed expert $e$ for all future stages, then the limit of her expected average payoff will be $\mu_e$.  Therefore, a potential approach to reliable learning against an open-ended partner is to eventually commit to following the best expert so far for all future stages.  We will show, however, that such \emph{passive} strategies can always suffer large regret against certain open-ended strategies.

\begin{definition}
    \label{def:passive}
    We say that a learning algorithm $\mathcal{A}$ is \emph{passive almost surely} if, for any game $G$ and set of experts $E$, when Alice plays $\mathcal{A}_{G,E}$ against any strategy $\phi$, we have that
    \begin{equation}
        \label{eqn:passive}
        \text{Pr}\left\{ \exists s, \forall t > s, a_t = a_s \right\} = 1
    \end{equation}
\end{definition}

If Alice's strategy $\mathcal{A}_{G,E}$ is passive a.s., then with probability $1$ there will be some stage $s$ after which Alice will follow the same expert in every stage.  In this section, we will show that for any passive algorithm $\mathcal{A}$ Alice could implement, we can construct a strategy $\phi$ for the coordination game such that Alice will suffer non-zero regret for playing $\mathcal{A}_{G,E}$ against $\phi$.  Our main tool for dealing with passive algorithms will be Lemma~\ref{lem:tau_delta}, which captures the intuition that if $\mathcal{A}$ is passive a.s., then there must be some finite number of stages within which Alice will commit to always following a single expert with high probability.

\begin{lemma}
    \label{lem:tau_delta}
    Assume that, when Alice plays strategy $\mathcal{A}_{G,E}$ against Bob's strategy $\phi$, we have
    \begin{equation}
        \text{\emph{Pr}}\left\{ \exists s, \forall t > s, a_t = a_s \right\} = 1 - \gamma
    \end{equation}
    for some $\gamma \in [0,1)$. Then for any $\delta > 0$, there exists $\tau$ s.t.
    \begin{equation}
        \label{eqn:tau_delta}
        \text{\emph{Pr}}\left\{ \exists s > \tau, a_s \neq a_\tau \right\} \leq \gamma + \delta
    \end{equation} % should this inequality be strict? - double check this, what is needed for the subsequent proofs
\end{lemma}

The proof of Lemma~\ref{lem:tau_delta} is provided in Appendix~\ref{apx:tau_delta}.  Note that Lemma~\ref{lem:tau_delta} is more general than we need for the passive case, and in Section~\ref{general_case} we will use it to prove bounds that apply to any algorithm.  For passive algorithms, we can use Lemma~\ref{lem:tau_delta} with $\gamma = 0$ to prove the following:

\begin{theorem}
    \label{thm:passive_regret}  % simplify the definition of passive algorithms
    Assume that algorithm $\mathcal{A}$ is passive a.s., then for any $N \geq 1$, there exists an $N \times N$ game $G$ and set of experts $E$ such that, for any $\delta > 0$, we can construct a strategy $\phi$ such that
    \begin{equation}
        \label{eqn:passive_regret}
        \mathcal{R}^{\text{\emph{adapt}}}_{G,E}(\mathcal{A}_{G,E}, \phi) \geq \frac{N - 2}{N} - \delta,
    \end{equation}
    with $\phi$ open-ended w.r.t. $E$ in $G$.
\end{theorem}

\begin{proof}
Let $G$ be the $N \times N$ coordination game, and let $E$ be the set of fixed-action experts.  Let $\phi^{\text{uni}}$ be the uniform random strategy.  We fix an arbitrary value of $\delta > 0$.  Let $C_s$ denote the event that Alice's actions have converged by stage $s$.  Since $\mathcal{A}$ is passive a.s., by Lemma~\ref{lem:tau_delta} there exists $\tau$ such that when Alice plays $\mathcal{A}_{G,E}$, we have $\text{Pr}_{\phi^{\text{uni}}}\{ \neg C_{\tau} \} \leq \gamma + \delta$ (with $\gamma=0$ since $\mathcal{A}$ is passive a.s.).  Henceforth, let $\tau$ denote the value satisfying Equation~\ref{eqn:tau_delta} for our choice of $\delta$.

We next define the \emph{switching} strategies.  For action $e \in [1, N]$, the switching strategy $\phi^{\tau}_e$ is defined as
\begin{equation}
    \phi^{\tau}_e(h) = 
    \begin{cases}
        e & \vert h \vert \geq \tau \wedge a(h) = e \\
        a \sim [1,N] & \text{otherwise}
    \end{cases}
\end{equation}
The strategy $\phi^{\tau}_e$ follows the uniform strategy for the first $\tau$ stages, and then continues to select actions randomly at each stage unless Alice took action $e$ in the previous stage, in which case $\phi^{\tau}_e$ takes action $e$.  Note that for any $e$, $\phi^{\tau}_e$ is open ended w.r.t. $E$, with $\mu_e = 1$, and $\mu_{e'} = \frac{1}{N}$. The latter is due to the fact that when Alice plays any action other than $e$, $\phi^{\tau}_e$ will sample its next action uniformly at random.

We next define the \emph{convergence probabilities} $p_e$ as
\begin{equation}
    p_e = \text{Pr}_{\phi^{\text{uni}}}\left\{ a_{\tau} = e \big\vert C_{\tau}\right\},
\end{equation}
that is, $p_e$ is the probability that the Alice converges to action $e$ (when playing against $\phi^{\text{uni}}$), given that her actions converge by stage $\tau$.  We then note that
\begin{equation}
    \text{Pr}_{\phi^{\tau}_e} \! \left\{C_{\tau} \! \wedge \! a_{\tau} \! \neq \! e \right\}  \! = \! \text{Pr}_{\phi^{\text{uni}}} \! \left\{C_{\tau} \! \wedge \! a_{\tau} \! \neq \! e \right\}.
\end{equation}
If Alice converges to $e' \neq e$ by stage $\tau$, then for each subsequent stage $s > \tau$, Bob's action distributions will be the same under both strategies, that is, $\phi^{\text{uni}}(h_s) = \phi^{\tau}_e(h_s)$.  So all finite histories $h$ will have equal probabilities under both measures. For the same reason
\begin{equation}
    \lim_{s \rightarrow \infty} \! \text{E}_{\phi^{\tau}_e} \! \left[ \frac{1}{s} \! \sum^{s}_{n=1} \! G(a_n, b_n) \! : \! C_{\tau} \! \wedge \! a_{\tau} \! \neq \! e \right] \! = \! \frac{1}{N}.
\end{equation}
We can then see that
\begin{align}
    \text{Pr}_{\phi^{\tau}_e} \! \left\{ \neg C_{\tau} \! \vee \! a_{\tau} \! = \! e \right\}  \! &= \! 1 - \text{Pr}_{\phi^{\tau}_e} \! \left\{C_{\tau} \! \wedge \! a_{\tau} \! \neq \! e \right\} \\
    &= \! 1 - \text{Pr}_{\phi^{\text{uni}}} \! \left\{C_{\tau} \! \wedge \! a_{\tau} \! \neq \! e \right\} \\
    = \! 1 -& \text{Pr}_{\phi^{\text{uni}}} \! \left\{ a_{\tau} \! \neq \! e \big\vert C_{\tau} \right\} \text{Pr}_{\phi^{\text{uni}}} \! \left\{ C_{\tau} \right\} \\
    = 1 -& (1 - p_e) [1 - \text{Pr}_{\phi^{\text{uni}}} \! \left\{ \neg C_{\tau} \right\}] \\
    \leq 1 -& (1 - p_e) [1 - \gamma - \delta] \\
    \leq \gamma +& \delta + p_e
\end{align}
Define $\tilde{e} \in \argmin_{e \in E}p_e$.  Since $\sum^{N}_{e=1}p_e = 1$, we have that $p_{\tilde{e}} \leq \frac{1}{N}$.  As the maximum average payoff is $1$, we have
\begin{align}
    V_{G}(\mathcal{A}_{G,E}, \phi^{\tau}_{\tilde{e}}) &\leq \gamma \! + \! \delta \! + \! p_{\tilde{e}} \! + \! \frac{1}{N}\text{Pr}_{\phi^{\tau}_{\tilde{e}}} \! \left\{C_{\tau} \! \wedge \! a_{\tau} \! \neq \! \tilde{e} \right\} \\
    &\leq \gamma \! + \! \delta \! + \! p_{\tilde{e}} \! + \! \frac{1}{N} \\
    &\leq \gamma \! + \! \delta \! + \frac{2}{N}
\end{align}
Here, we have decomposed the expected average payoff into the case where Alice converges to a fixed action $a \neq \tilde{e}$, for which we showed that the payoff is at most $\frac{1}{N}$, and the case where she fails to do so (where we can only use the trivial bound of $1$).  Finally, we have that
\begin{align}
    \mathcal{R}^{\text{adapt}}_{G,E}(\mathcal{A}_{G,E}, \phi^{\tau}_{\tilde{e}}) &= \mu_{\tilde{e}} - V_{G}(\mathcal{A}_{G,E}, \phi^{\tau}_{\tilde{e}}) \\
    &\geq 1 -\left[\gamma + \delta + \frac{2}{N}\right] \\
    &= \frac{N - 2}{N} - \gamma - \delta \label{eqn:gamma_passive}
\end{align}
where setting $\gamma = 0$ gives the desired result.
\end{proof}

\subsection{Active Algorithms}
\label{active_algorithms}

As we saw in the previous section, a learning algorithm that eventually commits to a single expert risks missing opportunities for a larger payoff.  In this section we show that the alternative, following an active learning strategy that never stops exploring, can fail in a different way.

\begin{definition}
    \label{def:active}
    We say that a learning algorithm $\mathcal{A}$ is \emph{active almost surely} if, for any game $G$ and set of experts $E$, when Alice plays $\mathcal{A}_{G,E}$ against any strategy $\phi$, we have that
    \begin{equation}
        \label{eqn:active}
        \text{Pr}\left\{ \forall s, \exists t > s,  a_t \neq a_s \right\} = 1
    \end{equation}
\end{definition}

If Alice's strategy $\mathcal{A}_{G,E}$ is active a.s., then for any stage $s$ where the she follows expert $e_s$, there will be a stage $t > s$ where she follows a different expert $e_t \neq e_s$.  An active strategy will never commit to a fixed expert indefinitely, and must eventually deviate from its current expert.

The issue is that if Bob is aware of Alice's strategy, he can predict at which stage Alice is likely to switch from her current expert.  Bob can then require that Alice continue to follow her current expert past this stage in order to learn anything about that expert's long-term payoffs.  Again, let $\phi^{\text{uni}}$ be the uniformly random strategy.  The following lemma shows that we can always predict, with high confidence, when Alice will deviate from her current expert.

\begin{lemma}
    \label{lem:sigma_delta}
    Assume that, when Alice plays strategy $\mathcal{A}_{G,E}$ against Bob's strategy $\phi$, we have
    \begin{equation}
        \text{\emph{Pr}}\left\{\forall s, \exists t > s, a_t \neq a_s \right\} = \gamma,
    \end{equation}
    for some $\gamma > 0$. Let $C$ denote the event that Alice's behavior eventually converges to some fixed action. For any history $h$ with positive probability under $\mathcal{A}_{G,E}$ and $\phi$, and any $\delta > 0$, there exists $\sigma$ such that
    \begin{equation}
        \label{eqn:sigma_delta}
        \text{\emph{Pr}}\left\{ \forall s \in [1, \sigma], a_{\vert h\vert} \! = \! a_{\vert h \vert + s} \;\big\vert\; h \wedge \neg C \right\} \leq \delta % Ideally, swap the form of this equation
    \end{equation}
\end{lemma}

The proof of Lemma~\ref{lem:sigma_delta} is analogous to that of Lemma~\ref{lem:tau_delta}, and is provided in Appendix~\ref{apx:sigma_delta}.  Using Lemma~\ref{lem:sigma_delta}, we can now prove the following theorem:

\begin{theorem}
    \label{thm:active_regret}
    Assume that algorithm $\mathcal{A}$ is active a.s., then for any $N \geq 1$, there exists an $N \times N$ game $G$ and set of experts $E$ such that, for any $\delta > 0$, we can construct a strategy $\phi$ such that
    \begin{equation}
        \label{eqn:active_regret} % There must have been an error here - should be minus \delta
        \mathcal{R}^{\text{\emph{adapt}}}_{G,E}(\mathcal{A}_{G,E}, \phi) \geq \frac{N - 1}{N} - \delta
    \end{equation}
    with $\phi$ open-ended w.r.t. $E$ in $G$.
\end{theorem}

Before getting into the proof of Theorem~\ref{thm:active_regret}, we briefly note that the theorem shows how the assumption of a finite evaluation horizon can be exploited.  Imagine Alice is using the \emph{strategic experts algorithm}~\cite{farias2003experts}.  Each time Alice selects an expert to evaluate, she also selects the time-horizon over which she will evaluate that expert.  Because Bob has complete knowledge of Alice's strategy, no matter how large a horizon she selects, Bob can always ensure that Alice learns nothing about the final performance of the current expert over this horizon.  Theorem~\ref{thm:active_regret} alone is sufficient to show that the algorithms presented in~\cite{farias2003experts},~\cite{even2005pomdp}, and~\cite{poland2005universal} can suffer positive adaptive regret against an open-ended partner, as they are each active almost surely. 

\begin{proof}
From the Bob's perspective, we can divide the learning process into \emph{intervals} $i$, such that a new interval starts every time Alice switches to following a different expert ($i=0$ starts at the beginning of the game).  Let $h^i$ be the history up to the stage $s_i$ at which the $i$th interval started.  Assuming that the learner's behavior does not converge (the event $\neg C$), Lemma~\ref{lem:tau_delta} means for each interval $i$, there must exist $\sigma_i$ such that, conditioned on $h^i$, the probability that the Alice will switch experts (and interval $i$ will end) by stage $s_i + \sigma_i$ is greater that $1- \delta$ (for any $\delta > 0$).  

To show that learning is impossible, we define a strategy $\phi_{\mathcal{A}}$\footnote{\citep{chang2001exploiter} define a similar predictive exploiter strategy for hill-climbing learners.} that, in each interval $i$, follows the uniform strategy $\phi^{\text{uni}}$ for the first $\sigma_i$ steps, and then follows expert $e_i$ that the learner is following in $i$ until the interval ends.  Given a sequence $\{\delta_i\}_{i \in \mathbb{N}}$, we define each $\sigma_i$ such that
\begin{equation}
    \text{Pr}\left\{ \forall s \in [1, \sigma_i], a_{\vert h\vert} = a_{\vert h\vert + s} \vert h^i \wedge \neg C \right\} \leq \delta_i
\end{equation} % Rationalize indices later
Let $\iota(s)$ be the interval in which stage $s$ occurs (note that both $\iota(s)$ and $\sigma_i$ are random variables on $(\Omega, \mathcal{F})$).  Clearly $\phi_{\mathcal{A}}$ is open ended, because if we commit to following any expert $e$ from stage $s$, then after at most $\sigma_{\iota(s)}$ stages, $\phi_{\mathcal{A}}$ will switch to $e$ as well, and so for all $e \in E$, $\mu_e = 1$.

Let $R^{s}_{i}$ denote the total payoff received in interval $i$, truncated to stage $s$ if $s$ occurs during or before interval $i$ (so that $R^{s}_{i} = 0$ if $s_i > s$).  Let $D_s$ be the event that, for all intervals $i$ such that $s_i \leq s$, either $s_{i+1} \leq s_i + \sigma_i$, or $s < s_i + \sigma_i$. More intuitively, $D_s$ is the event that, up to stage $s$, Alice has never followed a single expert for longer than Bob predicted she would.  We then see that
\begin{align}  % Could be a problem using the limit rather than the limit infimum
    V_{G}(\mathcal{A}_{G,E}, \phi_{\mathcal{A}}) &= \liminf_{n \rightarrow \infty} \text{E} \left[ \frac{1}{n} \sum^{n}_{s=1}G(a_s,b_s)\right] \\
    \leq (1-\gamma) + \gamma &\liminf_{n \rightarrow \infty} \text{E} \left[ \frac{1}{n} \sum^{n}_{s=1}G(a_s,b_s) \; : \; \neg C \right] \label{eqn:converge_decomp} \\
    = (1-\gamma) + \gamma &\liminf_{n \rightarrow \infty} \text{E} \left[ \frac{1}{n} \sum^{\iota(n)}_{i=1} R^{n}_{i} \; : \; \neg C \right] \\
    \leq (1-\gamma) + \gamma&\liminf_{n \rightarrow \infty} \text{E} \left[ \frac{1}{n} \sum^{\iota(n)}_{i=1} R^{n}_{i} : D_n \wedge \neg C \right]  \\
    + \text{Pr}\{ \neg & D_n \;\big\vert\; \neg C \} \text{E} \left[ \sum^{\iota(n)}_{i=1} R^{n}_{i} : \neg D_n \wedge \neg C \right] \label{eqn:violate_decomp} \\
    \leq (1-\gamma) + \gamma &\frac{1}{N} + \gamma\liminf_{n \rightarrow \infty} \text{Pr}\{ \neg D_n \;\big\vert\; \neg C\}
\end{align}
In Equation~\ref{eqn:converge_decomp} we are decomposing the expected payoff into the case where Alice's behavior does and does not converge.  In Equation~\ref{eqn:violate_decomp}, we are decomposing the expectation into two cases: 1) the case $D_n$ where there is no interval up to stage $n$ in which Alice followed $e_i$ for more that $\sigma_i$ stages, and 2) the case $\neg D_n$ where there is such an interval.  In the former case, Bob only ever plays the uniform base strategy $\phi^{\text{uni}}$, yielding an expected average payoff of $\frac{1}{N}$.  In the later case, Alice could observe and adapt to $e^*$, and so the maximum average payoff is $1$.  To bound $\liminf_{n \rightarrow \infty} \text{Pr}\{ \neg D_n \;\big\vert\; \neg C \}$, note that
\begin{align}
    \text{Pr}\{ \neg D_n \big\vert \neg C \} &\leq \text{Pr}\left\{ \bigcup^{n}_{i=1} \{ s_{i+1} > s_i + \sigma_i \} \big\vert \neg C \right\} \\
    &\leq \sum^{n}_{i=1} \text{Pr}\{ s_{i+1} > s_i + \sigma_i \;\big\vert\; \neg C \} \\
    &\leq \sum^{n}_{i=1} \delta_i
\end{align}
where the first inequality uses the fact that there can be at most $n$ intervals observed in the first $n$ stages, and so if we have not observed a violation of the predicted deviation horizon within the first $n$ intervals, we will not have observed it within the first $n$ stages.  Setting $\delta_i = \frac{\delta}{2^{i+1}}$ gives us $\liminf_{n \rightarrow \infty} \text{Pr}\{ \neg D_n \;\big\vert\; \neg C \} \leq \delta$.  The use of an exponentially decaying parameter $\delta_i$ means that Alice faces a receding evaluation horizon as learning progresses, and so with high probability $\mathcal{A}_{G,E}$ will only ever observe the uniform random strategy $\phi^{\text{uni}}$.  Finally, we have that
\begin{align}
    \mathcal{R}^{\text{adapt}}_{G,E}(\mathcal{A}_{G,E}, \phi_{\mathcal{A}}) &= 1 - V_{G}(\mathcal{A}_{G,E}, \phi_{\mathcal{A}}) \\
    &\geq 1 - (1 - \gamma) - \gamma \left[\frac{1}{N} + \delta \right] \\
    &= \gamma \left[\frac{N - 1}{N} - \delta \right] \label{eqn:gamma_active}
\end{align}
where setting $\gamma = 1$ gives us the desired result.
\end{proof}

\subsection{Proof of Theorem~\ref{thm:impossibility}}
\label{general_case}

Most learning algorithms described in the literature fall into one of the two special classes described above.  It is entirely possible, however, to define a learning algorithm that can be both active and passive, and so we must address this general case.  Given a passive algorithm $\mathcal{A}$, and an active algorithm $\mathcal{A}'$, we can define a new algorithm which starts by randomly choosing whether to use $\mathcal{A}$ or $\mathcal{A}'$ for the rest of the interaction.  Such an algorithm could either converge or explore forever, each with probability $\frac{1}{2}$.

To lower-bound the regret of such \emph{mixed} learning strategies, we simply combine the partner strategies used for the active and passive cases.  Note that Equations~\ref{eqn:gamma_passive} and~\ref{eqn:gamma_active} both depend on the probability $\gamma$ that Alice's behavior fails to converge when played against the uniform random strategy $\phi^{\text{uni}}$.  Therefore, for any algorithm $\mathcal{A}$, we can construct a strategy $\phi$ for the $N \times N$ coordination game, such that
\begin{equation}
    \label{eqn:mixture}
    \mathcal{R}^{\text{adapt}}_{G,E}(\mathcal{A}_{G,E}, \phi) \! \geq \! \max \{ \! \frac{N \!-\! 2}{N} \! -\! \gamma \! - \! \delta, \gamma\! \left[\frac{N \!-\! 1}{N}\! - \!\delta \right] \! \}
\end{equation} 
We can simplify the result by replacing the second term with $\gamma \left[\frac{N - 2}{N} - \delta \right]$.  As the first term decreases linearly with $\gamma$, and (for $N \geq 3$ and $0 < \delta < \frac{N-2}{N}$) the second term increases linearly with $\gamma$, we can minimize Equation~\ref{eqn:mixture} over possible values of $\gamma$ by setting
\begin{equation}
    \gamma = \left[\frac{N-2}{N} -\delta \right]\left( 1 + \frac{N-2}{N} - \delta\right)^{-1}
\end{equation}
the value at which both bounds are equal.  Finally, plugging this value into Equation~\ref{eqn:mixture} we get
\begin{align}
    \mathcal{R}^{\text{adapt}}_{G,E}(\mathcal{A}_{G,E}, \phi) \! &\geq \! \left[ \frac{N \!\! - \! 2}{N} \! - \! \delta \right]^2 \!\!\! \left(\! 1 + \frac{N \!\! - \! 2}{N} \! - \! \delta \!\right)^{-1} \\
    &\geq \frac{1}{2} \left[\frac{N - 2}{N} - \delta \right]^2
\end{align}
which completes the proof. $\square$

\subsection{Open-Ended Regret}
\label{open_ended_regret}

Rather than explicitly requiring that Bob will forgive any initial mistakes Alice might make, a more general solution is to redefine the regret itself so that it does not penalize Alice for making such mistakes.  We can do this by defining Alice's regret for playing strategy $\pi$ instead of strategy $e \in E$  in terms of the expected average payoff that $e$ can guarantee starting from any history $h$.

\begin{definition}
    \label{def:open_ended_regret}
    In a repeated game $G$, if Bob follows strategy $\phi$, Alice's \emph{open-ended regret}\footnote{\citep{powers2005finite} suggest a similar notion of regret as an alternative to finite history-dependence constraints.} for playing strategy $\pi$ w.r.t. a set of experts strategies $E$ is:
    \begin{equation}
        \label{eqn:open_ended_regret}
        \mathcal{R}^{\text{oe}}_{G,E}(\pi, \phi) = \max_{e \in E}\inf_{h \in \mathcal{H}}V_{G}(e, \phi \vert h) - V_{G}(\pi, \phi)
    \end{equation}
\end{definition}

The first term is a lower bound on the expected payoff expert $e$ can receive if an adversary were to select Alice's actions for a finite number of initial stages.  If Bob's strategy $\phi$ can enter an unrecoverable state, the adversary will force him into this state before control is handed over to $e$.  This value is useful because it tells us how well Alice could eventually do once she has identified a good response to $\phi$, no matter how her initial exploratory behavior has affected Bob's behavior.  Note that $\mathcal{R}^{\text{oe}}_{G,\phi}(\pi, \phi)$ can be negative even when $V_{G}(e, \phi) \geq V_{G}(\pi, \phi)$, as the open-ended regret accounts for worst-case histories that may never be reached.  We now have the following corollary to Theorem~\ref{thm:impossibility}:

\begin{corollary}
    \label{cor:open_ended_regret}
    For any $N \geq 1$, there exists an $N \times N$ game $G$ and experts $E$ such that, for any learning algorithm $\mathcal{A}$, and any $\delta > 0$, we can construct a strategy $\phi$ such that
    \begin{equation}
        \label{eqn:corollary}
        \mathcal{R}^{\text{\emph{oe}}}_{G,E}(\mathcal{A}_{G,E}, \phi) \geq \frac{1}{2}\left[\frac{N-2}{N} - \delta \right]^2
    \end{equation}
\end{corollary}

\begin{proof}
Consider the game $G$, Bob's strategy $\phi$, and experts $E$ which satisfy Equation~\ref{eqn:impossibility} for $\mathcal{A}_{G,E}$ and some $\delta > 0$.  $\phi$ is open-ended w.r.t. $E$, and so for all experts $e \in E$, $V_{G}(e,\phi \vert h) = \mu_{e}$ for all $h \in \mathcal{H}$.  Therefore, for all $e \in E$
\begin{equation}
    \inf_{h \in \mathcal{H}}V_{G}(e, \phi \vert h) = \mu_{e}
\end{equation}
We then have that 
\begin{align}
    \mathcal{R}^{\text{oe}}_{G,E}(\mathcal{A}_{G,E}, \phi) &= \max_{e \in E} \mu_{e} - V_{G}(\mathcal{A}_{G},\phi) \\
    &= \mathcal{R}^{\text{adapt}}_{G,E}(\mathcal{A}_{G,E}, \phi) \\
    &\geq \frac{1-\delta}{2}\left[\frac{N-2}{N} - \delta \right]^2
\end{align}
for any $N \geq 1$, which is the desired result.
\end{proof}

\subsection{Computability}

A potential limitation of this these results is that when $\mathcal{A}$ is neither passive nor active almost surely, the predictive strategy $\phi_{\mathcal{A}}$ may not be Turing computable, even when $\mathcal{A}_{G,E}$ itself can be implemented as a probabilistic Turing machine~\cite{arora2009complexity}.  If $\mathcal{A}$ is active a.s., we could implement $\phi_{\mathcal{A}}$ by iteratively computing the distribution over paths of play over the next $n,n+1,n+2,\ldots$ stages until the probability that $\mathcal{A}$ has switched experts is greater than $1 - \delta_{\iota(s)}$.  When $\mathcal{A}$ can decide at random whether it is ever going to commit to an expert, from some histories the probability of deviation may never reach $1 - \delta_{\iota(s)}$, and the simulation process may never terminate.  Therefore, Theorem~\ref{thm:impossibility} only applies if we admit non-computable strategies.  It remains an open question whether there exists a computable learning algorithm $\mathcal{A}$ that can guarantee no adaptive regret against all computable, open-ended partner strategies.

\section{Rational Adaptation}
\label{discussion}

Our results highlight the fundamental limitation of approaching multi-agent learning as a non-stationary, single-agent learning problem.  Even when we can assume that Bob's computational constraints mean that his strategy must admit a bounded evaluation horizon, Alice's own computational constraints may prevent her from implementing a no-regret algorithm over a long-enough timescale to exploit this horizon.  We suggest instead constraining Bob's ability to adapt to Alice's behavior by requiring that any adaptive strategy be ``rational'' with respect to Bob's beliefs about Alice's strategy.  Here we briefly discuss two potential models of such \emph{rational adaptation}, and their implications for learning.

\paragraph{Computational Rationality}

As strategies which condition on Alice's behavior will generally be more difficult to implement than those that do not, we can imagine that Bob will only choose an adaptive strategy if he expects it to outperform any non-adaptive strategy.  We can formalize this idea by assuming that Alice and Bob's strategies are implemented by finite state machines (or Turing machines) $\mathcal{M}_{\pi}$ and $\mathcal{M}_{\phi}$, and replacing the repeated game with the \emph{machine game}~\cite{shoham2008multiagent} $G^{M}$, where $G^{M}(\mathcal{M}_{\pi}, \mathcal{M}_{\phi}) = V_G(\pi, \phi)$. Let $\rho$ be Bob's \emph{belief} over Alice's possible strategies.  We can say that $\phi$ is \emph{computationally rational} w.r.t. $\rho$ if for any $\phi' \neq \phi$, either $G^{M}(\rho, \mathcal{M}_{\phi}) \geq G^{M}(\rho, \mathcal{M}_{\phi'})$, or $G^{M}(\rho, \mathcal{M}_{\phi}) = G^{M}(\rho, \mathcal{M}_{\phi'})$ and $\vert \mathcal{M}_{\phi} \vert \leq \vert \mathcal{M}_{\phi'} \vert$, where
\begin{equation}
    G^{M}(\rho, \mathcal{M}_{\phi}) = \text{E}_{\mathcal{M}_{\pi} \sim \rho}[G(\mathcal{M}_{\pi}, \mathcal{M}_{\phi})]
\end{equation}
and $\vert \mathcal{M}_{\phi} \vert$ is the size of the machine that implements $\phi$.  $\phi$ is not only a best-response to $\rho$, but is the least complex strategy out of the best-responses to $\rho$.  Applied to Example~\ref{exp:grim_trigger}, if $\rho(a^1) = 1$, then the grim-trigger strategy $\phi_1$ is not computationally rational, as it is more complex than the strategy that always plays $b^1$, with greater payoff.  If $\rho(a^1) = \rho(a^2) = \frac{1}{2}$ however, the strategy $\phi$ which adapts to Alice's is rational, even if it is more complex than any stationary strategy.

\paragraph{Cognitive Hierarchies}

In place of, or in addition to computational constraints on Bob's strategy, we might also place constraints on Bob's beliefs about Alice's strategy.  Under the \emph{cognitive hierarchies} model~\cite{camerer2004hierarchies}, a level-$k$ rational strategy is a best-response to some belief over level-$0,\ldots,k$ strategies.  Applied to repeated games, level-$0$ strategies could consist of all stationary strategies, with level-$1$ strategies adapting their partner's stationary strategy, and higher-level strategies reasoning about their partner's adaptive behavior.   Rather than bounding the memory or evaluation horizon of Bob's strategy, we might instead bound the maximum level $k$ of his strategy in the cognitive hierarchy.   This assumption of \emph{hierarchical rationality} admits complex, adaptive strategies, while potentially ruling out pathological strategies which are not a best-response to any set of strategies Alice could plausibly follow.  We leave as a goal for future work the development of learning algorithms with regret guarantees under the assumption that Bob's strategy is hierarchically or computationally rational.

\section{Conclusion}

In this work we have shown that there is no algorithm which can be guaranteed to learn to cooperate with an arbitrary partner in a repeated game, even when this partner is always willing to forgive mistakes made by the learner.  By providing a more complete understanding of the limits of the subjective approach to multi-agent learning, this work serves as a starting point for the exploration of new theoretical frameworks that specifically account for the intentions behind other agents' adaptive behavior.  Ultimately, we hope this line of research will lead to novel algorithms that can reliably learn cooperate with humans and other artificial agents.  Most importantly, we hope to motivate the development of learning algorithms that are robust to both adaptive and non-adaptive partners, without relying on unrealistic assumptions about these agents.

\section*{Acknowledgements}

This research was (partially) funded by the Hybrid Intelligence Center, a 10-year programme
funded by the Dutch Ministry of Education, Culture and Science through the Netherlands
Organisation for Scientific Research, \href{https://hybrid-intelligence-centre.nl}{https://hybrid-intelligence-centre.nl}, grant number
024.004.022.

\newpage
\bibliography{references}
\bibliographystyle{icml2022}

%%%%%%%%%%%%%%%%%%%%%%%%%%%%%%%%%%%%%%%%%%%%%%%%%%%%%%%%%%%%%%%%%%%%%%%%%%%%%%%
%%%%%%%%%%%%%%%%%%%%%%%%%%%%%%%%%%%%%%%%%%%%%%%%%%%%%%%%%%%%%%%%%%%%%%%%%%%%%%%
% APPENDIX
%%%%%%%%%%%%%%%%%%%%%%%%%%%%%%%%%%%%%%%%%%%%%%%%%%%%%%%%%%%%%%%%%%%%%%%%%%%%%%%
%%%%%%%%%%%%%%%%%%%%%%%%%%%%%%%%%%%%%%%%%%%%%%%%%%%%%%%%%%%%%%%%%%%%%%%%%%%%%%%
\newpage
\appendix
\onecolumn

\section{Proof of Proposition~\ref{prop:open_ended_flexible}}
\label{apx:open_ended_flexible}

Let $r > \frac{1}{4}$, $c$ and $\mu_e$ (for all $e \in E$) be values satisfying Equation~\ref{eqn:flexibility} for Bob's strategy $\phi$ in the repeated game $G$.  For any expert strategy $e \in E$, history $h \in \mathcal{H}$, and integer $s$, we have that
\begin{align}
    \text{E}_{a,b \sim e,\phi}\left[\frac{1}{s}\sum^{\vert h \vert + s}_{n = \vert h \vert + 1} G(a_n, b_n) \; : \; h\right] - \mu_e &= \text{E}_{a,b \sim e,\phi}\left[\frac{1}{s}\sum^{\vert h \vert + s}_{n = \vert h \vert + 1} G(a_n, b_n) - \mu_e \; : \; h\right] \\
    &\leq \text{E}_{a,b \sim e,\phi}\left[\left\vert \frac{1}{s}\sum^{\vert h \vert + s}_{n = \vert h \vert + 1} G(a_n, b_n) - \mu_e \right\vert \; : \; h \right] \\
    &\leq c s^{-r} \label{eqn:flex_upper}
\end{align}
Reversing this argument, we also have that
\begin{align}
    \text{E}_{a,b \sim e,\phi}\left[\frac{1}{s}\sum^{\vert h \vert + s}_{n = \vert h \vert + 1} G(a_n, b_n) \; : \; h\right] - \mu_e &\geq -\text{E}_{a,b \sim e,\phi}\left[\left\vert \frac{1}{s}\sum^{\vert h \vert + s}_{n = \vert h \vert + 1} G(a_n, b_n) - \mu_e \right\vert \; : \; h \right] \\
    &\geq -c s^{-r} \label{eqn:flex_lower}
\end{align}
Finally, using Equation~\ref{eqn:flex_upper} we get
\begin{align}
    V_{G}(e,\phi \vert h) - \mu_e &= \liminf_{s \rightarrow \infty} \text{E}_{a,b \sim e,\phi}\left[\frac{1}{s}\sum^{\vert h \vert + s}_{n = \vert h \vert + 1} G(a_n, b_n) \; : \; h\right] - \mu_e \\
    &\leq \liminf_{s \rightarrow \infty} c s^{-r} = 0
\end{align}
while using Equation~\ref{eqn:flex_lower} gives us
\begin{align}
    V_{G}(e,\phi \vert h) - \mu_e &\geq \liminf_{s \rightarrow \infty} -c s^{-r} = 0
\end{align}
Therefore, $V_{G}(e,\phi \vert h) - \mu_e = 0$ for any $e \in E$, and so $\phi$ is also open-ended with respect to $E$. $\square$

\section{Proof of Proposition~\ref{prop:fictitious_play_open_ended}}
\label{apx:fictitious_play_open_ended} % This whole section seems to be fouled up

Let $E$ be the set of $\vert A \vert$ fixed-action experts in $G$.  Let $\alpha_s$ be the empirical strategy up stage $s$, defined as $\alpha_s(a) = \frac{1}{s}\sum^{s}_{n=1}1_{\{ a_n = a \}}$.  For each expert $e \in E$, let $\mu_e = \max_{b \in B}G(e, b)$. Let $\phi_{\text{fp}}$ be the strategy that implements fictitious play.  We need to show that, for each $e \in E$, and for each history $h \in \mathcal{H}$, we have $\mu_e - V_G(e, \phi^{\text{fp}} \vert h) = 0$.  The main issue is showing that, for any initial history $h$, when Alice commits to expert $e$, Bob's fictitious play strategy will eventually play a best-response to $e$.  We define the $B(e) = \argmax_{b}G(e, b)$, and then define the minimum regret $\epsilon(e)$ as
\begin{equation}
    \epsilon(e) = \min_{b \not\in B(e)}\left[ \mu_e - G(e, b)\right]
\end{equation}
that is, $\epsilon(e)$ is the minimum regret over all of Bob's sub-optimal pure-strategy responses to $e$.  Assume that from stage $n = \vert h\vert$ Alice follows expert $e$ for all subsequent stages.  Then for any $s$ we have
\begin{align}
    \alpha_{\vert h \vert + s}(e) &= \frac{1}{\vert h \vert + s}\left[ s + \sum_{n=1}^{\vert h \vert}1_{\{a_n = e\}} \right] \\
    &\geq \frac{s}{\vert h \vert + s}
\end{align}
In stage $s > 1$, fictitious play plays Bob's best response to $\alpha_{s-1}$.  Therefore, for Bob to select a sub-optimal action $b \not\in B(e)$ in stage $s+1$, it will need to be the case that
\begin{equation}
    G(\alpha_{s}, b') - G(\alpha_{s}, b) \leq 0
\end{equation}
for all $b' \in B(e)$, where $G(\alpha_{s}, b)$ is shorthand for $\sum_{a}\alpha_{s}(a)G(a, b)$.  Starting from an initial history $h$, let Alice follow expert $e$ for $s = \left\lceil\frac{\vert H \vert}{\epsilon(e)}\right\rceil + 1$ stages.  We then have that
\begin{align}
    G(\alpha_{\vert h \vert + s}, b') - G(\alpha_{\vert h \vert + s}, b) &= \alpha_{\vert h \vert + s}(e)\left[G(e, b') - G(e, b)\right] + \sum_{a \neq e}\alpha_{\vert h \vert + s}(a)\left[G(a, b') - G(a, b)\right] \\
    &\geq \alpha_{\vert h \vert + s}(e)\epsilon(e) - \left[ 1 - \alpha_{\vert h \vert + s}\right] \\
    &\geq (1 + \epsilon(e)) \frac{s}{\vert h \vert + s} - 1 \\
    &= (1 + \epsilon(e)) \frac{\left\lceil\frac{\vert H \vert}{\epsilon(e)}\right\rceil + 1}{\vert h \vert + \left\lceil\frac{\vert H \vert}{\epsilon(e)}\right\rceil + 1} - 1 \\
    &> (1 + \epsilon(e)) \frac{\frac{\vert h\vert}{\epsilon(e)}}{\vert h \vert + \frac{\vert h\vert}{\epsilon(e)}} - 1 = 0 
\end{align}
This means that after Alice follows $e$ for $s = \left\lceil\frac{\vert H \vert}{\epsilon(e)}\right\rceil + 1$ steps, Bob's best response to $\alpha_{\vert h \vert + s}$ cannot be a suboptimal response to $e$, simply because $e$ will have too much mass under the resulting empirical distribution.  Bob will switch to some $b \in B(e)$ after no more than $s$ steps, and will never change back so long as Alice continues to play $e$.  Therefore, starting from any history $h \in \mathcal{H}$, for any $e \in E$, we have that $V_G(e, \phi^{\text{fp}}) = \mu_e$, concluding the proof. $\square$

\section{Proof of Lemma~\ref{lem:tau_delta}}
\label{apx:tau_delta}

Define the random variables $c_s$ on $(\Omega, \mathcal{F})$ as
\begin{equation}
    \label{eqn:convergence_indicator}
    c_s(\omega) = \left\{
    \begin{array}{ll}
        0, & \exists t > s, a_t(\omega) \neq a_s(\omega) \\
        1, & \text{otherwise}
    \end{array}\right.
\end{equation}
For each stage $s$, the variable $c_s$ indicates whether Alice has converged to a fixed action by this stage.  We need to ensure that each $c_s : \Omega \mapsto \Re$ is measurable w.r.t. $\mathcal{F}$.  We also need to show that the $\lim_{s \rightarrow \infty} c_s$ is measurable.  To see this, let $C_{s,t} \in \mathcal{F}$ be the event that Alice takes a fixed action from stage $s$ to at least stage $t$.  We have that $\{\omega \in \Omega : c_s(\omega) = 1\} = \bigcap^{\infty}_{t=s} C_{s,t}$, and $\{\omega \in \Omega : \lim_{s \rightarrow \infty} c_s(\omega) = 1\} = \bigcup^{\infty}_{s=1} \{\omega \in \Omega : c_s(\omega) = 1\}$.  These are countable intersections and unions, and so $\{\omega \in \Omega : c_s(\omega) = 1\} \in \mathcal{F}$ and $\{\omega \in \Omega : \lim_{s \rightarrow \infty} c_s(\omega) = 1\} \in \mathcal{F}$.  Therefore, each $c_s$, as well as $\lim_{s \rightarrow \infty} c_s$, is measurable w.r.t. $\mathcal{F}$.

Note that because $c_s \in \{0,1\}$, the indicator functions $1_{\{c_s = 0\}} = 1 - c_s$, and the indicator function $1_{\{\lim_{s \rightarrow \infty} c_s = 0 \}} = 1 - \lim_{s \rightarrow \infty} c_s$.  We then have that
\begin{align}
    \gamma &= \text{Pr}\left\{ \forall s, \exists t > s, a_t \neq a_s \right\} \\
    &= \text{Pr}\left\{ \lim_{s \rightarrow \infty} c_s = 0 \right\} \\
    &= \text{E}\left[1_{\{\lim_{s \rightarrow \infty} c_s = 0\}}\right] \\
    &= \text{E}\left[1 - \lim_{s \rightarrow \infty} c_s \right] \\
    &= \text{E}\left[\lim_{s \rightarrow \infty} 1 - c_s \right] \\
    &= \lim_{s \rightarrow \infty} \text{E}\left[1 - c_s \right] \\
    &= \lim_{s \rightarrow \infty} \text{E}\left[ 1_{\{c_s = 0\}} \right] \\
    &= \lim_{s \rightarrow \infty} \text{Pr}\left\{ c_s = 0 \right\}.
\end{align}
Since $c_s$ is bounded, we can exchange the limit and expectation by the dominated convergence theorem~\cite{bartle2014integrals}.  This means that for any $\delta > 0$ there exists $\tau$ such that
\begin{equation}
    \text{Pr}\left\{ c_\tau = 0 \right\} \leq \gamma + \delta
\end{equation}
Since $\{c_s = 0\} \Leftrightarrow \{\exists t > s, a_t \neq a_s\}$, this implies that the probability that learning has not converged by stage $\tau$ is less than to $\gamma + \delta$, and so this $\tau$ satisfies Equation~\ref{eqn:tau_delta} for the chosen value of $\delta$.  $\square$

\section{Proof of Lemma~\ref{lem:sigma_delta}}
\label{apx:sigma_delta}

The proof is analogous to that of Lemma~\ref{lem:tau_delta}.  First, recall from the previous proof that $\{\omega \in \Omega : \lim_{s \rightarrow \infty} c_s(\omega) = 1\} = C$ is measurable w.r.t. $\mathcal{F}$. 
% The next equation may be redundant, but keep for now
We then note that for any $h \in \mathcal{H}$ with $\text{Pr}\left\{ h \right\} > 0$,
\begin{equation}
\text{Pr}\left\{\forall s > \vert h \vert, \exists t > s, a_t \neq a_s \;\big\vert\; h \wedge \neg C \right\} = 1.
\end{equation}
which comes from the definition of $\neg C$.  For all $t \in \mathbb{N}$, we define a sequence of random variables $d^{t}_s$ on $(\Omega, \mathcal{F})$ such that
\begin{equation}
    d^{t}_s(\omega) =
    \begin{cases}
        0 & \forall n \in [t, t + s], a_t(\omega) = a_n(\omega) \\
        1 & \text{otherwise}
    \end{cases}
\end{equation}
where $d^{t}_s = 1$ indicates that Alice will switch actions within the next $s$ stages after stage $t$.  We then have that
\begin{align} % Try to fix the indices here - not critical
    1 &= \text{Pr}\left\{\forall s > \vert h \vert, \exists t > s, a_t \neq a_s \;\big\vert\; h \wedge \neg C \right\} \\
    &=\text{Pr}\left\{ \lim_{s \rightarrow \infty} d^{\vert h \vert}_s = 1 \;\big\vert\; h \wedge \neg C \right\} \\
    &= \lim_{s \rightarrow \infty} \text{Pr}\left\{ d^{\vert h \vert}_s = 1 \;\big\vert\; h \wedge \neg C \right\}
\end{align}
where once again we use the fact that each $d^{\vert h \vert}_s$ and $\lim_{s \rightarrow \infty} d^{\vert h \vert}_s$ are bounded, combined with the dominated convergence theorem.  This means that for any $\delta > 0$, and any $h \in \mathcal{H}$ with $\text{Pr}\left\{ h \right\} > 0$, there exists $\sigma$ such that
\begin{equation}
    \text{Pr}\left\{ d^{\vert h \vert}_{\sigma} = 1 \;\big\vert\; h \wedge \neg C \right\} > 1 - \delta
\end{equation}
Since $\{ d^{t}_s \!\! = \!\! 1 \} \! \Leftrightarrow \! \{\exists n \!\in\! [t,t+s], a_t \!\!\neq\!\! a_n\}$, this $\sigma$ must satisfy Equation~\ref{eqn:sigma_delta} for the chosen value of $\delta$. $\square$

%%%%%%%%%%%%%%%%%%%%%%%%%%%%%%%%%%%%%%%%%%%%%%%%%%%%%%%%%%%%%%%%%%%%%%%%%%%%%%%
%%%%%%%%%%%%%%%%%%%%%%%%%%%%%%%%%%%%%%%%%%%%%%%%%%%%%%%%%%%%%%%%%%%%%%%%%%%%%%%

\end{document}